\documentclass{IEEEtran}
\usepackage{mathtools}
\usepackage{footmisc}
\def\changesHilighted{true}
\usepackage{amsmath}

\usepackage{graphicx,epstopdf}
\usepackage{cite}
\usepackage{color}
\usepackage{amsfonts}
\usepackage{stfloats}
\usepackage{bm}
\usepackage[utf8]{inputenc}
\usepackage[]{footmisc}
\usepackage{amsthm}
\usepackage{algorithmic}
\usepackage{algorithm}
\usepackage{amsmath,amsthm,amsfonts,amssymb,amscd}
\usepackage{lastpage}
\usepackage{enumerate}
\usepackage{fancyhdr}
\usepackage{mathrsfs}
\usepackage{xcolor}
\usepackage{graphicx}
\usepackage[english]{babel}
\usepackage{hyperref}
\usepackage{bookmark}
\usepackage{graphicx}  
\usepackage{amssymb}
\usepackage{amsmath}
\usepackage{amsthm}

\newtheorem{example}{Example}
\newtheorem{problem}{Problem}
\newcommand{\norm}[1]{\lVert #1 \rVert}
\newcommand{\abs}[1]{\left| #1 \right|}

\definecolor{dogwoodrose}{rgb}{0.84, 0.09, 0.41}

\definecolor{slateblue}{rgb}{0.42, 0.35, 0.8}
\ifnum\pdfstrcmp{\changesHilighted}{false}=0
    
\fi
\newcommand{\R}{\mathbb{R}}

\newcommand{\C}{\mathcal{C}}

\newtheorem{theorem}{Theorem}
\newtheorem{remark}{Remark}
\newtheorem{assumption}{Assumption}
\newtheorem{definition}{Definition}
\newtheorem{lemma}{Lemma}
\newtheorem{proposition}[theorem]{Proposition}

\def\BibTeX{{\rm B\kern-.05em{\sc i\kern-.025em b}\kern-.08em
    T\kern-.1667em\lower.7ex\hbox{E}\kern-.125emX}}
\usepackage{graphicx}
\usepackage[caption=false]{subfig}

\begin{document}
\title{Safety Critical Control for Nonlinear Systems with Complex Input Constraints
}
\author{Yaosheng Deng, Masaki Ogura, Yang Bai,  Yujie Wang,
and Mir Feroskhan
\thanks{This work was partially supported by JSPS KAKENHI Grant Number JP21H01352. Yang Bai and Mir Feroskhan are the co-corresponding authors.}
\thanks{Yaosheng Deng and Mir Feroskhan are with the School of Mechanical and Aerospace Engineering, Nanyang Technological University, Singapore, Singapore 639798 (e-mail:
yaosheng001@e.ntu.edu.sg; mir.feroskhan@ntu.edu.sg).}

\thanks{Yang Bai and Masaki Ogura are with the Graduate School of Advanced Science and Engineering, Hiroshima University, Hiroshima, Japan (e-mail: yangbai@hiroshima-u.ac.jp; oguram@hiroshima-u.ac.jp).}
\thanks{Yujie Wang is with the Department of Mechanical Engineering, University
of Wisconsin-Madison, Madison, WI, USA (ywang2835@wisc.edu).}
}
\maketitle

\begin{abstract}
In this paper, we propose a novel Control Barrier Function (CBF) based controller for nonlinear systems with complex, time-varying input constraints. To deal with these constraints, we introduce an auxiliary control input to transform the original system into an augmented one, thus reformulating the constrained-input problem into a constrained-output one. This transformation simplifies the Quadratic Programming (QP) formulation and enhances compatibility with the CBF framework. As a result, the proposed method can systematically address the complex, time-varying, and state-dependent input constraints. The efficacy of the proposed approach is validated using numerical examples.
\end{abstract}


\section{Introduction}
In practical control systems, input constraints commonly exist, arising from physical limitations such as actuator saturation, safety requirements, or energy restrictions. These input-constrained systems present unique challenges in control design, as conventional methods may yield infeasible or unsafe commands when such constraints are not explicitly considered~\cite{matsutani2010adaptive, mir,furtat2024nonlinear}. Ensuring system performance, and constraint satisfaction simultaneously requires tailored strategies.

To deal with this challenge, Model Predictive Control (MPC) has been a prominent method for managing constraints~\cite{intro-mpc}. Several studies have investigated its application in handling input constraints in nonlinear systems. 
However, nonlinear MPC requires solving a nonlinear programming problem, which is not always feasible for online applications due to the limitations of QP solvers in low-dimensional parameter spaces~\cite{soloperto2022nonlinear, ferreau2008online}. Alternatively, the reference governor (RG) approach~\cite{garone2015explicit} integrates input constraints into a well-designed nominal controller using QP. Despite its effectiveness, RG necessitates the computation of admissible sets, complicating its implementation~\cite{liu2018decoupled}. Barrier Lyapunov Function (BLF) based approaches have also been widely adopted to manage constraints in various nonlinear systems. For instance, BLF-based controllers have been proposed for systems with input saturation~\cite{li2020barrier, mousavi2023barrier}. However, BLFs primarily address time-varying constraints and often overlook the more complex scenario of state-dependent constraints. This focus on time-based constraints limits their applicability in systems where the state and environment can change unpredictably~\cite{blf-bad-1, li2024composite}. Furthermore, BLF methods typically require the reference trajectory to remain within the constraint set, adding complexity to the design process and potentially restricting system performance~\cite{blf-bad-2, blf-bad-3}.

Motivated by these limitations, recent studies have explored CBF as an alternative framework to systematically handle system constraints~\cite{ames2020integral,cortez2022safe,taylor2020adaptive}. In CBF-based approaches, two conditions are used to enforce output constraints. The first is the CBF condition, which ensures that the safe set remains invariant, and the second is the Control Lyapunov Functions (CLF) condition, which guarantees stability. 
For control affine system, these conditions are affine constraints. As a result, they can be combined into a single convex optimization problem that is solved via quadratic programming (QP), unlike other methods that are used for non-affine problems~\cite{dacs2025robust}. This CLF-CBF-QP framework yields globally optimal solutions~\cite{ames2016control, ames2019control}.

However, the application of CBF-based designs to systems with input constraints is limited. One method to address this is through integral control barrier functions (ICBFs)\cite{ames2020integral}. While promising, further theoretical investigation is needed to establish the feasibility issue of ICBF-based controllers\cite{cortez2022safe}, as highlighted in [Rem 4, 16]~\cite{ames2020integral}.
Another approach involves incorporating input saturation directly into the QP formulation. 
In~\cite{intro-multicbf-4}, input constraints are defined as one of the multiple CBF conditions in the QP formulation. Although this approach has been successful in certain specific models, introducing multiple constraints in the QP could potentially lead to infeasibility issues~\cite{intro-cbf-ic-1}. 
To address these challenges, several studies have proposed methods
based on specific assumptions. For example, in~\cite{intro-multicbf-2}, the authors assume that the safety regions of multiple CBFs do not conflict, which allows each CBF to be treated independently. However, this assumption is often unrealistic in practical scenarios.
In~\cite{intro-multicbf-3}, a multiple CBF-based approach for robot navigation is proposed, but it relies on a specifically structured environment. These assumptions can simplify the problem but do not fully resolve the underlying challenges of handling input constraints with CBFs. Consequently, managing input constraints in CBF-based control designs remains a complex and unresolved issue.

In this research, we propose a novel CBF-based scheme for input-constrained nonlinear systems, where constraint boundaries are related to both state and time. Instead of incorporating the input constraint directly into the QP formulation, we transform the constrained-input problem into a constrained-output one. 
This transformation aligns with the solid CBF framework~\cite{ames2019control} thereby simplifying the QP formulation and relaxing the non‑conflict assumptions required by previous CBF approaches (e.g.,~\cite{intro-multicbf-4,intro-cbf-ic-1,intro-multicbf-2,intro-multicbf-3}).
Specifically, we add an integrator into the feedback loop of the original system so that the original input becomes an output of the augmented system. This transformation allows us to apply CBF methods directly, ensuring that the input constraints are satisfied.
While this transformation simplifies the control problem and enhances compatibility with the CBF framework, it also introduces mismatched disturbance into the original system (see Section~\ref{section-main}). Inspired by adaptive-CBF (aCBF)~\cite{taylor2020adaptive}, we address this issue by approximating the time-varying disturbance using a weighted combination of basis functions. Update laws are designed to estimate these weights in real-time, thereby eliminating the disturbance’s impact on the control system. Our approach systematically mitigates the challenges posed by complex, time-varying input constraints, ensuring reliable operation under varying conditions. Additionally, it enhances system robustness and performance by employing an aCBF to handle system disturbance effectively.

The rest of the paper is organized as follows. In Section~\ref{section-2}, some notations and preliminaries are introduced. In Section~\ref{section-ps}, a safe constrained input problem is stated for an $n$th order nonlinear system, and a corresponding control algorithm with input constraints is developed based on the CBF technique in Section~\ref{section-main}. The proposed control
the algorithm is verified under simulations in Section~\ref{section-5}, and
finally, conclusions are drawn in Section~\ref{section-6}.
\section{Preliminary}\label{section-2}
In this section, the concepts of CLF and CBF are
reviewed, which are the main tools for our controller design.
\subsection{Notation}
We denote the set of real numbers by $\mathbb{R}$ and non-negative reals by $\mathbb{R}_+$. A continuous function $\alpha:[0,a)\to[0,\infty)$ is class- $\mathcal{K}$ for some $a>0$ if it is strictly increasing on the domain, and $\alpha(0)=0$. It is class-$\mathcal{K}_\infty$ if $\lim_{r\to\infty}\alpha(r)\to\infty$. The Euclidean norm of a vector $x\in\R^n$ is given by $\norm{x}=\sqrt{x^\top x}$. The Lie derivative of a
continuously differentiable function $h\colon \R^n\to \R$ with respect
to a Lipschitz continuous function $f\colon \R^n\to \R$ is $\mathcal{L}_fh(x)=\frac{\partial h}{\partial x}f(x)$. Let $x(t)$ be a real-valued function defined on $t\in[t_0,t_1]$. The supremum of $x(t)$ over the interval $[t_0,t_1]$, denoted by $x_{\text{sup}}(t)$, is defined as $x_{\text{sup}}(t)=\sup_{t\in[t_0,t_1]}x(t)$.  Similarly, the infimum of $x(t)$ over the interval $[t_0,t_1]$, denoted by $x_{\text{inf}}(t)$, is the greatest lower bound of $x(t)$ such that $x_{\text{inf}}(t)=\inf_{t\in[t_0,t_1]}x(t)$.
\subsection{CLF and CBF}
Consider the following control affine system: \cite{ames2019control}
\begin{equation}\label{ctrlaffinesys}
\dot{x}(t)=f(x(t))+g(x(t))u(t),
\end{equation}
where $x(t)=\begin{bmatrix}x_1(t),\ldots,x_n(t)\end{bmatrix}^\top\in\mathbb{R}^n$ is the state vector, $u(t)\in\R^m$ is a constrained control input, and $f:\mathbb{R}^n\to\mathbb{R}^n$ and $g:\mathbb{R}^n\to\mathbb{R}^{n\times m}\setminus\{{0}\}$ are smooth
continuous and locallly Lipschitz functions.
In the rest of the preliminary, we omit time $t$ for $x$ and $u$, provided no confusion arises.
\begin{definition}~\cite{taylor2020adaptive}\label{def clf}
Let $V\colon \R^n\to\R$ be a continuously differentiable, positive definite, and radially unbounded function.
Then $V(x)$ is a control Lyapunov function (CLF) for system~\eqref{ctrlaffinesys} if there exist $\alpha_1$, $\alpha_2$ and $\alpha_3\in\mathcal{K}_{\infty}$ such that:
\begin{align}
\alpha_1(\norm{x})\leq V(x)&\leq \alpha_2(\norm{x}),\\ 
\inf_{u\in\mathbb{R}^m}\left[\mathcal{L}_fV(x)+\mathcal{L}_gV(x)u\right]&\leq-\alpha_3(\norm{x}),
\end{align}
for all $x\in\R^n$, where $\mathcal{L}_fV({x}) = \frac{\partial V}{\partial x}f(x)$ and $\mathcal{L}_gV({x}) = \frac{\partial V}{\partial x}g(x)$ are the Lie derivatives.
\end{definition}
This definition means that there exists a set of stabilizing
controls that renders the origin globally asymptotically stable. This set is defined by
\begin{equation}\label{CLF}
K_{\text{clf}}(x)\!=\!\left\{u\in \mathbb{R}^{m}:\mathcal{L}_{f}V(x)\!+\!\mathcal{L}_{g}V(x)u{\le}\!-\!\alpha_3(\norm{x})\right\},\end{equation}
for all $x\in\R^n$.
Safety can be framed in the context of enforcing
invariance of a particular set of states. 
Consider control system~\eqref{ctrlaffinesys} and suppose
there exists a set $\mathcal{C}\subset \R^n$ defined as the 0-superlevel set of a continuously differentiable function $h:\R^n\to \R$, as follows:
\begin{equation}
\begin{aligned}
    \mathcal{C}=\{x\in\R^n:h(x)\geq0\}.
\end{aligned}
\end{equation}
The set $\mathcal{C}$ is referred to as the safe set, which we assume this set is closed, non-empty, and simply connected.
\begin{definition}\label{def Safety}
The set $\mathcal{C}$ is called forward controlled
invariant with respect to system~\eqref{ctrlaffinesys} if for every $x_0\in\mathcal{C}$,  there exists a control signal $u\colon [t_0,\infty)\to\R^m$ such that $x(t;t_0,x_0)\in\mathcal{C}$ for all $t\geq t_0$, where $x(t;t_0,x_0)$ denotes the solution of~\eqref{ctrlaffinesys} at time
$t$ with initial condition $x_0\in\R^n$ at $t=t_0$.
\end{definition}

\begin{definition}\label{def CBFs}
Let $h:\mathbb{R}^n\to\mathbb{R}$ be a continuously differentiable function that is used to define the safe set $\mathcal{C}\subset\R^n$ in Definition~\ref{def Safety}. Then $h$ is a CBF with (input) relative degree $1$ if the condition
\begin{equation}\label{lflg}
\sup_{{u}\in \R^m}\left[\mathcal{L}_fh({x})+\mathcal{L}_gh({x}){u}+\gamma_h h({x})\right]\geq0,
\end{equation}
is satisfied for all $x\in\R^n$. Given a CBF $h$, the set of all control values that satisfy~\eqref{lflg} is defined as
\begin{equation}\label{def kcbf}
    K_{\text{cbf}}(x) = \{u\in\R^m: \mathcal{L}_fh({x})+\mathcal{L}_gh({x}){u}+\gamma_h h({x})\geq 0\},
\end{equation}
for all $x\in\R^n$.
\end{definition}

  It was proven in~\cite{ames2019control}
that any Lipschitz continuous controller $u$ satisfying $u(t)\in K_{\text{cbf}}(x(t))$ for every $x\in \R^n$ guarantees the forward invariance of $\mathcal{C}$. The
provably safe control law is obtained by solving an online
quadratic program (QP) problem that includes the control
barrier condition as its constraint.
\subsection{Projection operator}
The projection operator of two vectors is defined by~\cite{FAT1,FAT2} as follows:
\begin{equation}
\begin{aligned}[b]
&\mathrm{Proj}({x},{y},l(x)) \\
&=\! \begin{cases}
    {y}\!-\!l({x})\frac{\nabla l({x})\nabla l({x})^\top}{\|\nabla l({x})\|^2}{y},\! & \!\mathrm{if~}l({x})\!>\!0\! \wedge{y}^\top\nabla l({x})\!>\!0, \\
    {y}, & \mathrm{otherwise},
   \end{cases}
\end{aligned}
\end{equation}
for all $x\in\R^n$ and $y\in\R^n$. $l({x})$ is convex function
defined as
\begin{equation}
l({x})=\frac{{x}^\top{x}-\bar{x}^2}{2\eta\bar{x}+\eta^2},
\end{equation}
where $\bar{x}$ and $\eta$ are constants.
\begin{lemma}~\cite{FAT2}\label{lemma x-x*}
Let ${x}^*\in\R^n$ such that $l(x)\leq 0$. Let $x^*=2x$, then
\begin{equation}
-{x}^\top(\mathrm{Proj}({x},{y},l(x))-{y})\leq0.
\end{equation}
Let $x^*=0$, then 
\begin{equation}
{x}^\top(\mathrm{Proj}({x},{y},l(x))-{y})\leq0.
\end{equation}
\end{lemma}
\section{Problem statement}\label{section-ps}
In this section, we propose a CBF-based controller design to ensure safety for the system~\eqref{ctrlaffinesys} with input constraints and disturbances. 

Firstly, we consider a nonlinear control-affine dynamical system with the unknown external disturbance
\begin{align}\label{sys1}
    \dot{x}(t)&=f(x(t))+g(x(t))u(t) + d_x(t),
\end{align}
where, $d_x(t) \in\mathcal{D}\subset \mathbb R^n$
is an unknown external disturbance of time $t$ such that $d_x(t)\in \mathcal D$ for all $t$ for a subset $\mathcal D$ of $\R^n$. We denote the initial state and control input of the system at time $t=0$ by $x_0$ and $u_0$, respectively, i.e., $x(0)=x_0$, $u(0)=u_0$. 
We introduce $\kappa(x(t),t)$, a time-varying continuous scalar function that depends on $x$ and $t$, as the input constraint:
\begin{equation}\label{ic}
    \norm{u(t)}\leq \kappa(x(t), t),
\end{equation} for  all $t\geq 0$. 
The magnitude of the control input is
expected to be kept within limits imposed by the actuator’s saturation constraints.  However, current BLF-based methods commonly involve the feasibility conditions on constraint set.
Specifically, when the time-varying saturation includes an unfeasible region will pose difficulty for control safety, as in Example~\ref{exp1}:
\begin{example}\label{exp1}
We consider a simple but representative case of~\eqref{sys1}:
\begin{align}\label{sys example}
        \dot x_1(t) &= x_2(t), \\ 
    \dot x_2(t) &= u(t),
\end{align}
where $u\in\mathcal{U}$ is the control input subjected to a closed control constraint set defined as
\begin{equation}\label{U-example}
    \mathcal{U}\!=\!\{u\colon\!\R_+\!\to\mathbb{R}^m, \underline{k}_l(t)\!\leq\! u(t)\!\leq\! \overline{k}_h(t) \text{ for all } t\!\geq \!0\},
\end{equation}
 where $\underline{k}_l:\R_+\to\R$ and $\overline{k}_h:\R_+\to\R$ are the lowest and highest levels of input constraint such that $\underline{k}_l(t)<\overline{k}_h(t)$ for all $t\geq 0$. We designed a symmetric time-varying constraint as
\begin{align}
\underline{k}_l(t)&=-\sin(t)-1,\\    \overline{k}_h(t)&=\sin(t)+1.
\end{align}
For the system in Example~\ref{exp1}, to implement the input constraints via barrier-function-based methods, we refer to a solid barrier function in~\cite{mousavi2023barrier} as
\begin{equation}\label{bf1}
\mathcal{B}(t) \!=\! b(u;\underline{k}_l,\overline{k}_h)\!=\!\operatorname{log}\left(\frac{\overline{k}_h(t)}{\underline{k}_l(t)}\frac{\underline{k}_l(t)\!-\!u(t)}{\overline{k}_h(t)\!-\!u(t)}\right),
\end{equation}
where $b:\R\to\R$ is the barrier function defined on $(\underline{k}_l, \overline{k}_h)$, as it is obvious to see, if $u$ approaches the boundaries of the permitted range $(\underline{k}_l, \overline{k}_h)$, $\mathcal{B}$ will approach infinity, i.e., $\lim_{u\to \underline{k}_l^+}b(u;\underline{k}_l,\overline{k}_h)=-\infty$, or $\lim_{u\to \overline{k}_h^-}b(u;\underline{k}_l,\overline{k}_h)=+\infty$.
Note that $\inf_{t\geq 0}\{\overline{k}_h(t)\}=0$ and $\sup_{t\geq 0}\{\underline{k}_l(t)\}=0$, one can always find $t_0$ such that 
\begin{equation}\label{t0}
\exists t_0>0:(u_{\sup}(t_0)=0\lor u_{\inf}(t_0)=0),
\end{equation}
and we define the set $T_0$ that $t$ satisfies~\eqref{t0} as 
\begin{equation}
T_0 = \{t\in\mathbb{R}^+:u_{\sup}(t)=0\lor u_{\inf}(t)=0\}.
\end{equation}
Therefore, for $t\geq 0$, $t\notin T_0$,  $\mathcal{B}$ is bounded, then input constraints~\eqref{U-example} are automatically satisfied. However, for $t\in T_0$, then $\overline{k}_h(t) = 0$ or $\underline{k}_l(t) = 0$, and obviously $\mathcal{B}(t)$ diverges. This demonstrates that the barrier function $\mathcal{B}$ cannot enforce forward invariance of the input safety set $\mathcal{U}$ under the given input constraints.
\end{example}
To address such an unsafe condition, and guarantee the input constraint, we define an input constraint safe set 
 for system~\eqref{sys1} based on the CBF
technique. One defines a Lipschitz continuous function $h$ as a barrier function 
\begin{equation}\label{bf}
h(x,u, t) = -u(t)^{\top}u(t) + \kappa^{2}(x, t),
\end{equation}
and to guarantee the input constraint, we let a safe set $\mathcal{C}_u$ for actual control input $u$ as
\begin{equation}\label{safe set}
    \mathcal{C}_u = \{u\colon\!\R_+\!\to\mathbb{R}^m, h(x,u,t)\geq 0\}.
\end{equation}
The FAT is an effective tool for dealing with control systems with time-varying unknown disturbances. For instance, let $d(t)$ be an unknown time-varying function in a control system. One can utilize weighted basis functions to represent $d(t)$ at each time instant, as shown in~\cite{FAT1, FAT2}:
\begin{equation}\label{FAT}
{d}(t)=\sum_{i=1}^\infty{w}_i\psi_{h,j}(t),
\end{equation}
where $w_i$ denotes an unknown constant vector (weight) and $\psi_{h,j}(t)$ is the basis function to be selected. It is a common practice to design an update law that approximates the weights $w_i$ to mitigate the impact of $d(t)$ on the control system. Several candidates for the basis function $\psi_{h,j}(t)$ in\eqref{FAT} can be chosen to approximate the nonlinear functions. In this paper, we select the same form of $\psi_{h,j}$ as in\cite{FAT2}. This preliminary framework sets the stage for the design of the specific update law, which will be detailed in the subsequent sections.

\begin{assumption}\label{assumption bound}
The FAT of $d(t)$ in~\eqref{FAT} 
   satisfies $\norm{w_i}\leq \bar{w}_i$ for all $i$, $\bar{w}_i$ is a known positive constant.
\end{assumption}

Now, we can state the main objective of this paper:
\begin{problem}
Given the system~\eqref{sys1}, design a state feedback controller $u$ such that for
any $u_0\in \C_u$, the closed-loop trajectories of~\eqref{sys1} satisfy $\lim_{t\to\infty}x(t)= 0$ and $u(t)\in\C_u$ for all $t\geq0$. 
\end{problem}

\section{CBF-based Controller Design}\label{section-main}
In this section, we design our CBF-based controller for input-constrained system~\eqref{sys1}. First, we introduce an auxiliary control input to transform the original system into an augmented system, thereby converting the original constrained input problem into a constrained output problem. Next, we propose an aCBF-based method to ensure the safety of the system with input constraints. Finally, we demonstrate that combining this safety controller with a stabilizing nominal control law through a quadratic program achieves the desired behavior, as outlined in our problem statement.
\subsection{Auxiliary transformation}\label{section-main-1}
To provide time-varying bounds on the actual control variable 
$u$, it is natural to place an integrator in the feedback path to augment the system's output as the input of an auxiliary system. Specifically, by introducing an integrator for the control input $u$, the original first-order system in~\eqref{sys1} is transformed into a second-order system, where the time derivative of $u$ is treated as a new auxiliary input $v$. However, a potential disadvantage of this augmentation is the explicit introduction of mismatched disturbances. Consequently, the augmented system can now be described as:
\begin{align}\label{sys1 casestudy}
    \dot x(t) &= f(x(t))+g(x(t))u(t) + d_x(t), \\
    \dot u(t) &= v(t) + d_u(t),
\end{align}
where $d_x(t)\in \mathcal{D}\subset \mathbb R^n$ and $d_u(t)\in \mathcal{D}\subset \mathbb R^n$ are unknown disturbances of time $t$, and $v(t)\in\R^m$ is an auxiliary input defined as:
\begin{equation}\label{v=phi+mu}
    v(t) = \phi(t) + \mu(t) ,
\end{equation}
where $\phi(t)\in\R^m$ is the auxiliary dynamics~\eqref{sys1 casestudy}, and $\mu(t)\in\R^m$ is the safety controller represents the difference between auxiliary input $v$ and nominal control $\phi$. We refer to system~\eqref{sys1 casestudy} as the nominal system when $\mu(t) = 0$ for all $t\geq 0$.
\begin{remark}
The disturbance in the system~\eqref{sys1 casestudy} will always be regarded as sensor faults polluting all the states~\cite{remark-why-d}.
 The pollution caused by such sensor faults cannot be separated from the real signal, thus being mixed into the feedback signal and processed by the algorithm. Thus we address such a scenario that all the states including $u$  are polluted due to sensor faults coinciding in each system state, which is of theoretical and practical significance. 
\end{remark}
The following proposition gives CLF for system~\eqref{sys1 casestudy}. Explicit time dependence of
variable $t$ is omitted in the rest of this paper when it is clear from the context.
\begin{proposition}\label{lemma xiaozang}
    Suppose $\mu(t) = 0$ for all $t\geq 0$ in system~\eqref{sys1 casestudy}, and there exist a continuously differentiable function $V_0: \mathbb{R}^n \to \mathbb{R}_{\geq 0}$ and a legacy feedback controller $u_d(x, \hat{w}_x) \in \mathbb{R}^m$ for system~\eqref{sys1 casestudy}, where $\hat{w}_x$ is an update law designed later. If $u_d(0,0) = 0$ and 
    \begin{equation}\label{V0}
        L_{\tilde{f}_{clf}} V_0(x,\!  \hat{w}_x) \!+\! L_{g} V_0(x, 
\! \hat{w}_x) u_d(x, \hat{w}_x) \! \leq  \!\ \gamma_3({x}),
\end{equation}
        for all $x \in \mathbb{R}^n$, where $\gamma_1$ , $\gamma_2$, $\gamma_3$ are class $\mathcal{K}_{\infty}$ functions, and $\tilde{f}_{clf}$ is defined by
        \begin{equation}
            \tilde{f}_{clf}(x, \hat{w}_x) = f(x) + d_x - \hat{w}_x.
        \end{equation}
    Defining a function $V: \mathbb{R}^n \times \mathbb{R}^m\times \mathbb{R}^n \times \mathbb{R}^m \to \mathbb{R}_{\geq 0}$ as
    \begin{align}
        V(x, u, \hat{w}_x, \hat{w}_u) &\!=\! V_0(x, \hat{w}_x) \!+\! (d_u \!-\! \hat{w}_u)^\top (d_u \!-\! \hat{w}_u) \nonumber \\
        &\quad \!+\! (u \!-\! u_d(x, \hat{w}_x))^\top (u \!-\! u_d(x, \hat{w}_x)),
    \end{align}
    where $\hat{w}_u$ is another update law similar to $\hat{w}_x$. We further suppose that $u_d$ in~\eqref{V0} and $\phi$ in~\eqref{sys1 casestudy} can be designed such that 
    \begin{equation}\label{qp general CLF}
        \dot{V}(x, u, v, \hat{w}_x, \hat{w}_u) \leq -\gamma_3(\|{x}\|) - \gamma_4(\|u - u_d(x, \hat{w}_x)\|),
    \end{equation}
    where $\gamma_4$ is a class $\mathcal{K}_{\infty}$ function.
    Then $V$ in~\eqref{qp general CLF} is a CLF for system~\eqref{sys1 casestudy}. 
\end{proposition}

\begin{proof}
    The proof follows directly from the assumptions and the definition of CLF on Definition~\ref{def clf}. Since $V_0$ satisfies the given inequalities and $u_d$ stabilizes the system~\eqref{sys1 casestudy}, the constructed function $V$ inherits these properties, establishing $V$ as a control Lyapunov function for system~\eqref{sys1 casestudy}.
    Furthermore, we have:
    \begin{equation}\label{CLF-INEQUALITY-CONSTRAINT}
        \inf_{v \in \mathbb{R}^n}\! \dot{V}\!(x,\! u,\! v,\! \hat{w}_x,\! \hat{w}_u) \!<\! -\!\gamma_3(\|{x}\|) \!-\! \gamma_4(\|u \!-\! u_d(x,\! \hat{w}_x)\|),
    \end{equation}
    for all $x \neq 0$ and $u \neq u_d$. Hence, $V$ is a CLF for the system.
\end{proof}

    Suppose a valid control barrier function $h(u, \kappa)$ is associated with the input constraint set $\mathcal{C}_u$. Then from Definition~\ref{def CBFs} and Lemma~\ref{lemma xiaozang}, a safe CLF-CBF-QP-based optimization problem for system~\eqref{sys1 casestudy} could be  defined as follows:
    \begin{equation}\label{qp general}
    \begin{aligned}
        &\underset{\mu \in \mathbb{R}^m}{\min} \quad \| \mu \| \\
        &\mathrm{s.t.} \\ 
        &\dot{V}(x, u, v, \hat{w}_x, \hat{w}_u) \!<\! -\!\gamma_3(\|{x}\|)\! -\! \gamma_4(\|u - u_d(x, \hat{w}_x)\|), \\
        &L_f h(u, \kappa) \!+\! L_g h(u, \kappa) v \!\geq \!\gamma_h(h(u, \kappa)) ,
    \end{aligned}
    \end{equation}
    where $\gamma_h$ is a class $\mathcal{K}_{\infty}$ function ensuring the input constraint.

The following two steps will be introduced to derive the inequality constraints in~\eqref{qp general}.
Firstly, we design a nominal controller $\phi$ for the stability of the nominal system, as the CLF inequality constraint shown in~\eqref{qp general}. Then unifying this stability condition with CBF
safety condition~\eqref{safe set}, as the second inequality constraint in~\eqref{qp general}, then solved by QP optimization~\cite{ames2016control}. 
\subsection{CLF inequality constraint}\label{subsection-clf}
 To compensate for the effects of time-varying disturbances
$d_x$ and $d_u$ in system~\eqref{sys1 casestudy}, using FAT approach, the approximation of system~\eqref{sys1 casestudy} can be represented as
\begin{align}\label{sys2 casestudy}
\dot x &= f + g u + \sum_{i=1}^N  w_{x,i}\psi_{x,i} ,\\
\dot u & = \phi + \mu+\sum_{i=1}^N w_{u,i}\psi_{u,i},
\end{align}
where $N$ is the number of basis functions used in the approximation. $w_{x,i}$ and $w_{u,i}$ denotes the unknown constant vector, $\psi_{x,i}(t)$ and $\psi_{u,i}(t)$ are the basis functions to be selected. 

The following theorem shows that we can construct a feedback controller $\phi$ to locally achieve the CLF inequality constraint~\eqref{CLF-INEQUALITY-CONSTRAINT} which
stated in Proposition~\ref{lemma xiaozang} 

\begin{theorem}\label{thm-clf}
Define the nominal control $\phi$ in system~\eqref{sys2 casestudy} as
\begin{equation}\label{nominal phi}
  \begin{multlined}[b]  
    \phi = \frac{1}{g}\Big[-\dot f-\sum_{i=1}^N \big(\hat w_{u,i}\psi_{u,i}
    + \dot{\hat w}_{x,i}\psi_{x,i}+\hat w_{x,i}\dot{\psi}_{x,i}\big)\\
    \quad-\frac{c_u}{\theta_u}\Big(f+gu+\sum_{i=1}^N\hat w_{x,i}\psi_{x,i}
    +c_x\frac{x}{\theta_x}\Big)\\
    \quad-\dot{g}u-\frac{c_x}{\theta_x}(f+g u)\Big].
  \end{multlined}
\end{equation}

where $c_x$, $c_u$ and $\theta_x$, $\theta_u$ are positive constants, $\hat w_{x,i}$ and $\hat w_{u,i}$ are two update laws given by
\begin{align}\label{adaptive and update}
\dot{\hat w}_{x,i} &= \lambda_x^{-1}\psi_{x,i}x,\\
\dot{\hat{w}}_{u,i} &= \lambda_u^{-1}\psi_{u,i}s_u\nonumber\\
& = \lambda_u^{-1}\psi_{u,i}\left(f+g u+\sum_{i=1}^N \hat w_{x,i}\psi_{x,i} + c_x\dfrac{x}{\theta_x}\right).
\end{align}
Then, all closed-loop system signals in~\eqref{sys2 casestudy} are bounded and $\lim_{t\to \infty} x(t) = 0$.
\end{theorem}
\begin{proof}
To guarantee the stability of the nominal system, in the rest of this section, we assume $\mu(t) = 0$ for all $t>0$ in~\eqref{v=phi+mu}. We further define the sliding surface as
\begin{align}\label{slidesurface}
    s_x &= x - x_d, \\
    s_u &= f + g u - u_d,
\end{align}
where $x_d$ and $u_d$ represent the desired value of state $x$ and $u$ follows
\begin{align}\label{ud}
x_d &= 0,\\ 
u_d &= -\sum_{i=1}^N \hat w_{x,i}\psi_{x,i} - c_x\dfrac{s_x}{\theta_x}.
\end{align}
From~\eqref{slidesurface} we have
\begin{align}\label{slidingsurface dot}
\dot s_x &= (s_u + u_d) + d_x - \dot x_{d},\\
\dot s_u & = \dot f+\dot g u +g(v+d_u)- \dot u_d,
\end{align}
where, $x_{d}$ is the desired state of $x$, and for our control objective, we let $x_{d}(t)=0$ for all $t\geq 0$. We define 
\begin{equation}\label{bar du}
    \bar{d}_u = g d_u + \dot{\hat{d}}_x + \dfrac{c_x}{\theta_x}d_x,
\end{equation} 
and the derivative of $s_u$ in~\eqref{slidingsurface dot} is simplified as
\begin{equation}\label{su dot}
    \dot s_u = \dot g u+g\phi + \bar{d}_u - \ddot x_d - \dfrac{c_x}{\theta_x}(f+g u-\dot x_d).
\end{equation}
Using the function approximation technique given by~\eqref{ud}, \eqref{adaptive and update}, for~\eqref{su dot} and~\eqref{nominal phi}, one obtains 
\begin{equation}
\dot s_u = \sum_{i=1}^N(w_{u,i}-\hat w_{u,i}) \psi_{u,i} - c_u\dfrac{s_u}{\theta_u}.
\end{equation}

Let us design a Lyapunov
function candidate for the second order of the system~\eqref{sys2 casestudy} as
\begin{equation}
V_u \!= \!\dfrac{1}{2}\left(s_u^\top s_u \!+\! \lambda_u\sum_{i=1}^N(w_{u,i}\!-\!\hat w_{u,i})^\top (w_{u,i}\!-\!\hat w_{u,i})\right).
\end{equation}
Take time derivative of $V_u$ along the trajectory of $\dot s_u$ in~\eqref{slidingsurface dot} and we have
\begin{equation}\label{Vu dot 1}
\dot V_u = -c_2\dfrac{s_u^\top s_u}{\theta_u} + \sum_{i=1}^N(w_{u,i}-\hat w_{u,i})^\top (\psi_{u,i}s_u - \lambda_u \dot{\hat w}_{u,i}).
\end{equation}
Using the update law of $\dot{\hat w}_{u,i}$ in~\eqref{adaptive and update}, then~\eqref{Vu dot 1} yields
\begin{equation}\label{Vu dot 2}
\dot V_u  = -c_2\dfrac{s_u^\top s_u}{\theta_u},
\end{equation}
then~\eqref{Vu dot 2} implies
$s_u\in\mathcal{L}_2\cap \mathcal{L}_\infty$ and $w_{u,i}-\hat w_{u,i}\in \mathcal{L}_\infty$. Asymptotic convergence
of $s_u$ can thus be proved by using Barbalat’s lemma.

The results obtained above can be summarized as follows:
The output of system~\eqref{sys2 casestudy} converges to the boundary layer
by using the controller~\eqref{nominal phi} and update law~\eqref{adaptive and update} if sufficient
numbers of basis functions are used and the approximation
errors can be ignored. 

To prove the stability of the error signal $s_x$, let us define
the Lyapunov function candidate
\begin{equation}
V_x \!=\! \dfrac{1}{2}\Big(s_x^\top s_x \!+\! \lambda_x\sum_{i=1}^N(w_{x,i}\!-\!\hat w_{x,i})^\top (w_{x,i}\!-\!\hat w_{x,i})\Big).
\end{equation}
 The time derivative of $V_x$ is computed as
\begin{equation}\label{Vx dot 1}
\dot V_x\! =\! s_x^\top s_u\! - \!c_x\dfrac{s_x^\top s_x}{\theta_x}\! +\! \sum_{i=1}^N(w_{x,i}\!-\!\hat w_{x,i})^\top (\psi_{x,i}s_x \!-\! \lambda_x \dot{\hat w}_{x,i}).
\end{equation}
Using the update law of $\dot{\hat w}_{x,i}$ in~\eqref{adaptive and update}, the equation~\eqref{Vx dot 1} becomes
\begin{equation}\label{Vx dot 2}
\dot V_x = s_x^\top s_u - c_x\dfrac{s_x^\top s_x}{\theta_x}.
\end{equation}
Since $\dot V_u\leq 0$ implies $\abs{s_u(t)}\leq \abs{s_u(0)}$ for all $t>0$ and $\abs{s_u(t+T)}\leq \theta_u$ for some $T>0$, we may design $c_x$ as
\begin{equation}
    c_x = \theta_u + \delta, \quad \delta>0,
\end{equation}
so that \eqref{Vx dot 2} can be further derived to have 
\begin{equation}\label{Vx dot 3}
\begin{aligned}[b]
\dot{V}_{x}& =s_{x}^\top s_u-\left(\theta_n+\delta\right)\frac{s_{x}^2}{\theta_x}  \\
&\leqslant|s_x|\Big(|s_u(0)|-(\theta_{u}+\delta)\frac{|s_{x}|}{\theta_{x}}\Big).
\end{aligned}
\end{equation}
If 
\begin{equation}s_x\notin \left\{s\bigg|\abs{s}\leqslant\frac{|s_u(0)|\theta_x}{\theta_u+\delta}\right\},\end{equation}
then $\dot{V}_{x}\leq 0$, and
hence $s_x$ is bounded. This implies that before $s_u$ converges
to the boundary layer, $s_x$ is bounded. Once $\abs{s_u}\leq\theta_u$, there are three cases to be considered:

\textit{Case 1}: $s_x>\theta_x>0$.

From~\eqref{Vx dot 3}, we have
\begin{equation}\begin{aligned}
\frac{\dot{V}_{x}}{s_{x}}&\leqslant\theta_u-(\theta_u+\delta)\frac{s_{x}}{\theta_x}\leqslant-\delta \frac{s_{x}}{\theta_x},
\end{aligned}\end{equation}
which implies
\begin{equation}
    \dot{V}_{x} \leq -\delta \dfrac{s_x^2}{\theta_x}\leq 0.
\end{equation}

\textit{Case 2}: $s_x<-\theta_x<0$.

From~\eqref{Vx dot 3}, we have
\begin{equation}\begin{aligned}
\frac{\dot{V}_{x}}{s_{x}}&=s_u+(\theta_u+\delta)\frac{\abs{s_{x}}}{\theta_x}\geqslant-\delta \frac{\abs{s_{x}}}{\theta_x},
\end{aligned}\end{equation}
which implies
\begin{equation}
    \dot{V}_{x} \leq -\delta \dfrac{\abs{s_x}^2}{\theta_x}\leq 0.
\end{equation}

\textit{Case 3}: $\abs{s_x}\leq \theta_x$.

In this case, $s_x$ has already converged to the boundary layer, i.e. $s_x$ is bounded by $ \theta_x$.

From the above three cases, we know that once $s_u$ converges inside its boundary layer, $s_x$ is bounded and will also converge to its boundary layer. This gives boundedness of all signals and $s_x\in\mathcal{L}_2\cap \mathcal{L}_\infty$. Furthermore, $(w_{x,i}-\hat w_{x,i})\cap \mathcal{L}_\infty$, then asymptotic convergence of $s_x$ can thus be proved by using Barbalat’s lemma.
\end{proof}

Using nominal controller in~\eqref{nominal phi}, the approximation of $d_x$ and $d_u$ in~\eqref{sys2 casestudy} and auxiliary system in~\eqref{sys1 casestudy}, one yields the CLF inequality constraint in \eqref{qp general} as follows:
\begin{equation}\label{inequality-1}
  \begin{multlined}[b][.7\linewidth]  
    \mu^\top \Bigl(f+gu-\sum_{i=1}^N \hat w_{x,i}\psi_{x,i}
      -  c_x\dfrac{x}{\theta_x}\Bigr)\\
    -\dfrac{c_u}{\theta_u}\Bigl\lVert f+gu-\sum_{i=1}^N \hat w_{x,i}\psi_{x,i}
      - c_x\dfrac{x}{\theta_x}\Bigr\rVert^2 \leq 0.
  \end{multlined}
\end{equation}
\subsection{A safe controller design}
To compensate for the effects of unknown disturbance $d_u$ in system~\eqref{sys1 casestudy}, similar to the FAT approach in subsection~\ref{subsection-clf}, the auxiliary term in \eqref{sys1 casestudy} can be represented as
\begin{equation}\label{control system}
    \dot u = v + \sum_{j=1}^M{w}_{h, j}\psi_{h,j}(t) ,
\end{equation}
where $M$ is the number of basis functions used in the approximation, $w_{h, j}$ denotes an unknown constant vector, $\psi_{h,j}$
is the basis function to be selected.
\begin{assumption}\label{assumption bound2}
    The input constraint boundary $\dot \kappa$ is bounded such that  $\dot \kappa\leq \Pi_{\kappa}$, where  $\Pi_{\kappa}$ is a positive constants.
\end{assumption}
\begin{theorem}\label{thm-cbf}
By constructing the update laws $\hat w_{h, j}$ for the parameter estimation as
\begin{equation}\label{hat di dot}
\dot{\hat{w}}_{h, j}=\mathrm{Proj}\left(\hat w_{h, j},-\frac1{2Q_{j}}\left(\frac{\partial h}{\partial u}\right)\psi_{h,j}-\frac\varrho2\hat w_{h, j},l_{di}\right),
\end{equation}
where 
\begin{equation}
l_{w_{h, j}}(\hat w_{h, j})=\frac{\hat w_{h, j}^\top\hat w_{h, j}-\bar{w}_{h,j}^2}{2\nu_i\bar{w}_{h,j}+\nu_i^2},
\end{equation}
$\nu_i$ is a small constant, and 
\begin{equation}\label{Q}
Q_{j}\leq\frac{h(v(0))}{2N(\|\hat w_{h, j}(0)\|+\bar{w}_{h, j})^2},
\end{equation}
any Lipschitz continuous controller $v\in K_{\text{cbf}}(u, \hat w_{h, j})$ where 
\begin{equation}\begin{aligned}[b]\label{Kcbf}
K_{\text{cbf}}(u, \hat w_{h, j}) 
= & \left\{v\in\mathbb{R}^m\mid\left(\frac{\partial h}{\partial u}\right)^\top\sum_{i=1}^N \hat w_{h, j}\psi_{h,j}-\zeta\right. \\
&\left.+\frac{\varrho}{2}\bigg(h-\sum_{i=1}^NQ_{j}{\bar{{w}}_{h, j}}^2\bigg)\geq0\right\},
\end{aligned}\end{equation}
with $\zeta=\left\lVert\frac{\partial h}{\partial\kappa}\right\rVert \Pi_\kappa$, 
will guarantee the safety of $\C_u$ in regard to system~\eqref{control system}.
\end{theorem}
\begin{proof}
Define $\bar{h}$ as
\begin{equation}\label{h bar}
\bar{h}=h-\sum_{j=1}^MQ_{j}\tilde{w}_{h, j}^\top\tilde{w}_{h, j},
\end{equation}
where $\tilde{w}_{h, j}=w_{h, j}-\hat{w}_{h, j}$. To prove Theorem~\ref{thm-cbf}, one needs to show that $\bar{h}(t)\geq 0$ for all $t>0$, such that $h(t)\geq 0$ for all $t>0$ as required by~\eqref{safe set}. This property holds if $\dot{\bar{h}}$ can be expressed in the form of (or larger than) $-\lambda \dot{\bar{h}}$ where $\lambda>0$ with $\bar{h}(0)\geq 0$.

A reconstruction of $\dot{\bar{h}}$ to the form of $-\lambda \dot{\bar{h}}$ is demonstrated
as follows. With Assumption~\ref{assumption bound2},  $\dot{\bar{h}}$ is calculated as
\begin{equation}\label{h bar dot}
  \begin{aligned}[b]  
    \dot{\bar{h}} &
      =\Bigl(\frac{\partial h}{\partial u}\Bigr)^{\!\top}\dot u
      +\Bigl(\frac{\partial h}{\partial\kappa}\Bigr)^{\!\top}\dot{\kappa}
      -2\sum_{j=1}^M Q_j\,\tilde w_{h,j}^{\!\top}\dot{\tilde w}_{h,j}
    \\
    &\begin{multlined}[.85\linewidth]
       =\Bigl(\frac{\partial h}{\partial u}\Bigr)^{\!\top}
         \Bigl(v+\sum_{j=1}^M w_{h,j}\psi_{h,j}\Bigr)
       +\Bigl(\frac{\partial h}{\partial \kappa}\Bigr)^{\!\top}\dot{\kappa}
       \\+2\sum_{j=1}^M Q_j\,\tilde w_{h,j}^{\!\top}\dot{\hat w}_{h,j}
     \end{multlined}
    \\
    &\begin{multlined}[b][.85\linewidth]
       \ge\Bigl(\frac{\partial h}{\partial u}\Bigr)^{\!\top}
         \sum_{j=1}^M w_{h,j}\psi_{h,j}
       +\Bigl(\frac{\partial h}{\partial u}\Bigr)^{\!\top}v
       -\zeta
       \\+2\sum_{j=1}^M Q_j\,\tilde w_{h,j}^{\!\top}\dot{\hat w}_{h,j}
     \end{multlined}
  \end{aligned}
\end{equation}
As update law $\dot{\hat{w}}_{h, j}$ in~\eqref{h bar dot} is defined as~\eqref{hat di dot}, from Lemma~\ref{lemma x-x*}, one can see 
\begin{equation}\begin{aligned}[b]\label{dididot>}
&\tilde{w}_{h, j}^\top\dot{\hat{w}}_{h, j}\\
&\begin{multlined}[b][.85\linewidth]
     =(w_{h, j}-\hat{w}_{h, j})^\top   \\
\mathrm{Proj}\left(\hat{w}_{h, j},-\frac1{2Q_{j}}\left(\frac{\partial h}{\partial u}\right)\psi_{h,j}-\frac\varrho2\hat{w}_{h, j},l_{w_{h,j}}\right) 
\end{multlined}\\
&\geq-(w_{h,j}-\hat{w}_{h, j})^\top\left(\frac1{2Q_{j}}\bigg(\frac{\partial h}{\partial u}\bigg)\psi_{h,j}+\frac\varrho2\hat{w}_{h, j}\right).
\end{aligned}\end{equation}
Substituting~\eqref{dididot>} into~\eqref{h bar dot} yields
\begin{equation}\begin{aligned}[b]\label{dot bar h >}
&\begin{multlined}[b][.85\linewidth]
\dot{\bar{h}}\geq\left(\frac{\partial h}{\partial u}\right)^\top \sum_{j=1}^Mw_{h,j}\psi_{h,j} + \left(\frac{\partial h}{\partial u}\right)^\top v
-\zeta\\
-\sum_{j=1}^M{\bar{w}}_{h,j}^\top\left(\left(\frac{\partial h}{\partial u}\right)\psi_{h,j}+\varrho Q_{j}\hat{w}_{h,j}\right)\end{multlined}\\
&\begin{multlined}[b][.85\linewidth]\geq\left(\frac{\partial h}{\partial u}\right)^\top\left(\sum_{j=1}^M{\hat{w}}_{h,j}\psi_{h,j} + v\right)\\-\varrho\sum_{j=1}^{M}Q_{j}\tilde{w}_{h, j}^{\top}\hat{w}_{h,j}-\zeta. \end{multlined}\end{aligned}\end{equation}
Note that
\begin{equation}\label{didi<}
\tilde{w}_{h,j}^\top\hat{w}_{h,j}\!\leq\!\frac{w_{h,j}^\top w_{h,j}\!-\!\tilde{w}_{h,j}^\top\tilde{w}_{h, j}}2\!\leq\!\frac{\bar{w}_{h,j}^2\!-\!\tilde{w}_{h, j}^\top\tilde{w}_{h, j}}2.\end{equation}
The substitution of \eqref{didi<} into \eqref{dot bar h >} gives
\begin{equation}\begin{aligned}[b]\label{h bar dot > 2}
\dot{\bar{h}} &\geq \left(\frac{\partial h}{\partial u}\right)^\top v
+\frac{\varrho}{2}\left(\sum_{j=1}^MQ_{j}(\bar{w}_{h,j}^2-\tilde{w}_{h, j}^\top\tilde{w}_{h, j})\right)-\zeta \\
&\qquad+ \left(\frac{\partial h}{\partial u}\right)^\top\sum_{j=1}^M\hat{w}_{h,j}\psi_{h,j} \\
&=\Gamma+\frac12\varrho\bigg(\sum_{j=1}^MQ_{j}\tilde{w}_{h, j}^\top\tilde{w}_{h, j}\bigg), 
\end{aligned}\end{equation}
where
\begin{equation}\begin{aligned}\label{Gamma}
\Gamma \!=\!\left(\frac{\partial h}{\partial u}\right)^\top\!\!\left(\!v\!+\!\sum_{j=1}^M\hat{w}_{h,j}\psi_{h,j}\!\right)  \!\!-\!\frac\varrho2\bigg(\sum_{j=1}^MQ_{j}\bar{w}_{h,j}^2\bigg)\!-\!\zeta.
\end{aligned}\end{equation}
If $v$ in~\eqref{Gamma} is selected from ~\eqref{Kcbf}, the following condition is
satisfied $\Gamma \geq -\frac{\varrho}{2}h$, 
and thus, in virtue of~\eqref{h bar}, \eqref{h bar dot > 2} can be reexpressed as
\begin{equation}\dot{\bar{h}}\geq-\frac\varrho2\left(h-\sum_{j=1}^MQ_{j}\tilde{w}_{h,j}^\top\tilde{w}_{h, j}\right)=-\frac\varrho2\bar{h}.\end{equation}
In addition, as $\hat w_{h,j}$ are bounded by $\bar w_{h,j}$, $\bar h(0)$ satisfies
\begin{equation}\begin{aligned}[b]
\bar{h}(0)&\!=\!h(0)\!-\!\sum_{j=1}^M\!Q_{j}\left(w_{h,j}\!-\!\hat{w}_{h,j}(0)\right)^\top\!\left(w_{h,j}\!-\!\hat{w}_{h,j}(0)\right) \\
&\geq h(0)-\sum_{j=1}^MQ_{j}\Big(\bar{w}_{h,j}+\|\hat{w}_{h,j}(0)\|\Big)^2.
\end{aligned}\end{equation}
The selection of parameters $Q_{j}$ as~\eqref{Q} yields $\bar{h}(0)\geq 0$. According to the comparison lemma, we know $\bar{h}(t)\geq 0$ for all $t>0$, such that $h(t)\geq 0$ for all $t>0$ as desired.
\end{proof}
Finally, by using~\eqref{inequality-1} and \eqref{Kcbf} in Theorem~\ref{thm-cbf}, a safe controller is obtained by solving the following CLF-CBF-QP problem


\begin{IEEEeqnarray}{rCl}\label{QP}
    &&\min_{\mu}\|\mu\|^2
    \nonumber\\
&&\mathrm{s.t.~}
\nonumber\\
&&
\mu^\top\left(f+gu-\sum_{i=1}^N \hat w_{x,i}\psi_{x,i} -  c_x\dfrac{x}{\theta_x}\right)
\nonumber\\
&&\qquad\qquad-\dfrac{c_u}{\theta_u}\left\lVert f\!+\!gu\!-\!\sum_{i=1}^N \hat w_{x,i}\psi_{x,i} \!-\! c_x\dfrac{x}{\theta_x}\right\rVert^2\leq 0\label{equ qp1}, 
\\
&&
\!-\!2u^\top\Bigg(\sum_{j=1}^M\hat{w}_{h,j}\psi_{h,j} \!-\!c_x\dfrac{u}{\theta_x} \!-\! \dfrac{c_u}{\theta_u}\Big(u\!+\!\sum_{i=1}^N \hat w_{x,i}\psi_{x,i} 
\nonumber\\ 
&&\!+\!c_x\dfrac{x}{\theta_x}\Big)\!-\!\sum_{i=1}^N \hat w_{u,i}\psi_{u,i} \!-\! \mu\Bigg)
\!-\!2\Pi_\kappa\norm{\kappa}\nonumber\\
&&
\!+\! \frac{\varrho}{2}\left( h\!-\!\sum_{j=1}^MQ_{j}{\bar{{w}}_{h,j}}^2\right)\!\geq \!0\label{equ qp2}. 
\end{IEEEeqnarray}
\section{Case study}\label{section-5}
\begin{figure*}[!t]
  \centering
  \subfloat[Trajectory of $x$]{%
    \includegraphics[width=0.32\textwidth]{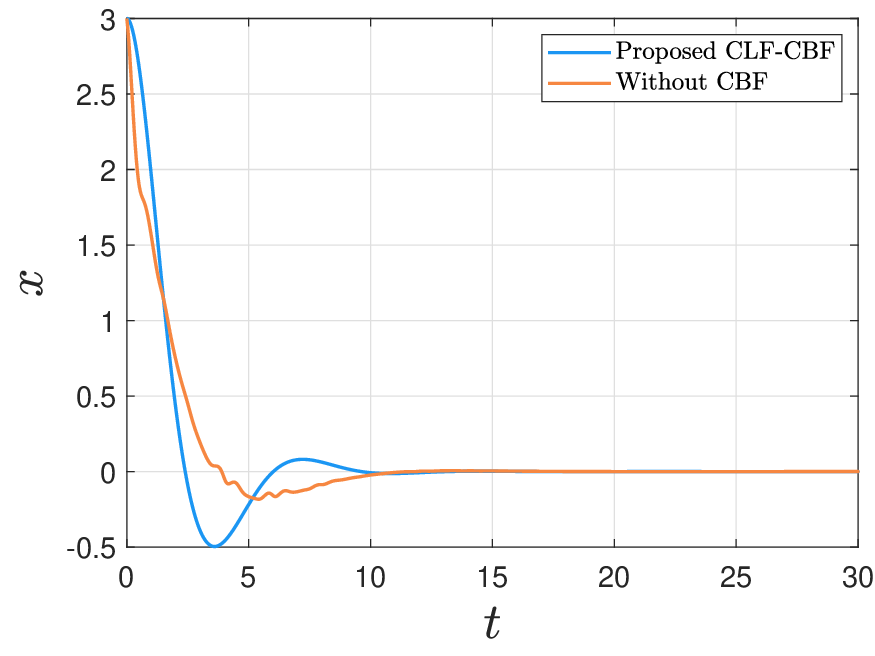}%
    \label{fig:subfig1}%
  }\hfill
  \subfloat[Control input $u$]{%
    \includegraphics[width=0.31\textwidth]{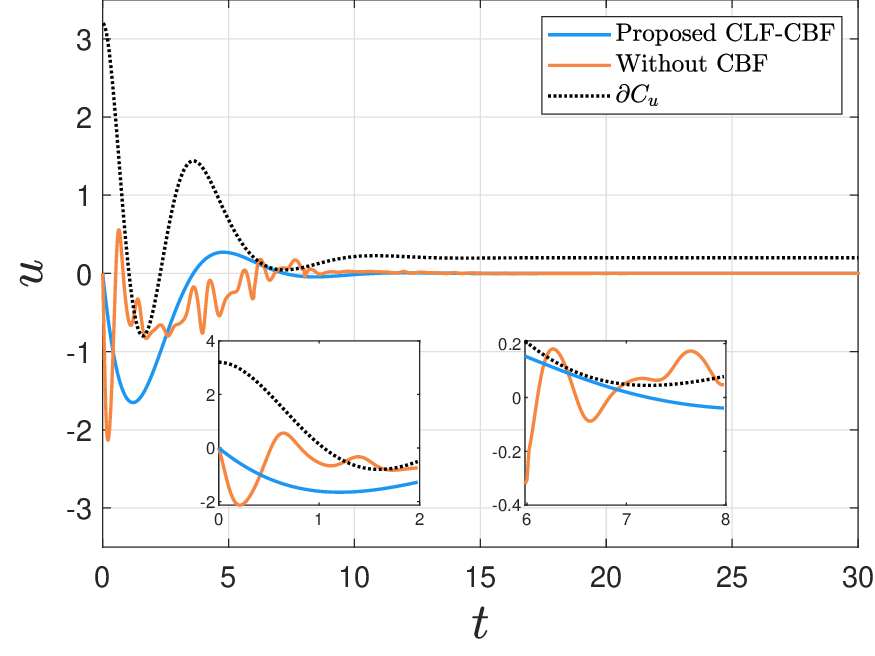}%
    \label{fig:subfig2}%
  }\hfill
  \subfloat[Barrier function $h$]{%
    \includegraphics[width=0.32\textwidth]{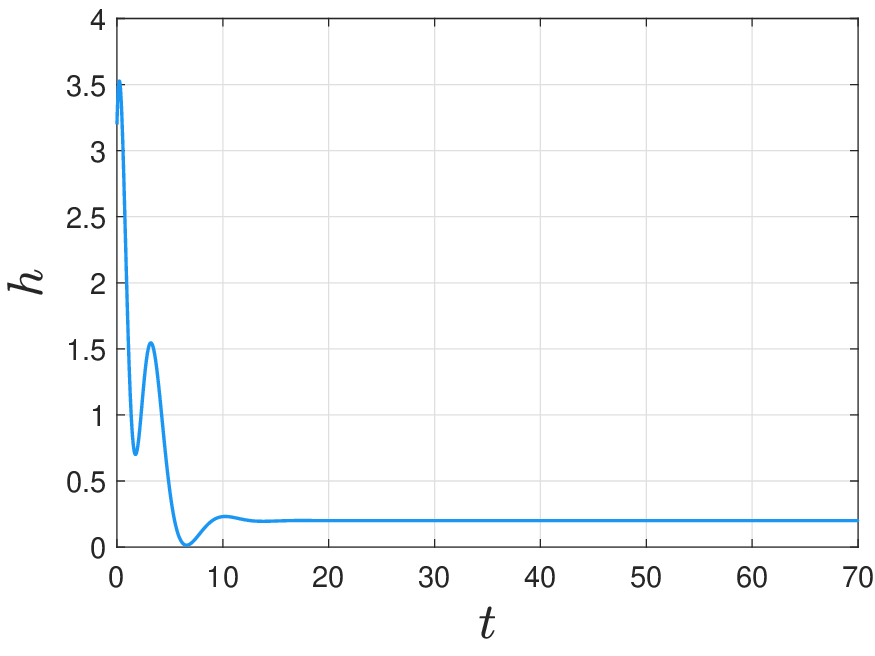}%
    \label{fig:subfig3}%
  }
  \vspace{-1.5mm}
  \caption{Case 1 simulation results. Subfigures (a) and (b) compare the system state trajectory $x$ and control input $u$ obtained using the proposed CLF-CBF (blue) and the nominal control $v=-x-x^2\text{sgn}(u)-u$ (magenta), respectively. (c) illustrates the barrier function $h$.}
  \label{fig:case1}
\end{figure*}
We first apply the proposed CBF-based controller to system~\eqref{sys example}.
We define the barrier function as $h(x,u) = \kappa - u$ for system~\eqref{sys example}, where $\kappa(x) = (x-1)^2-0.8$. Using system transformation in Section~\ref{section-main-1}, the auxiliary control input $v$ for system~\eqref{sys example} follows $\dot u = v$.
Our goal is to design the auxiliary control input $v$, such that $\lim_{t\to\infty} x(t)\to 0$ with $u\in\mathcal{C}_u$ for all $t\geq 0$ in system~\eqref{sys example}. To achieve this objective, one can design a nominal controller $\phi$ as $\phi = -x - x^2 \text{sgn}(u) -u$.
We set the initial conditions as $x(0)=3$ and $u(0)=0$, and set the constraint as $\kappa = (x-1)^2-0.8$ with a enough large constant $\Pi_{\kappa} = 15$ to satisfied  $\norm{\dot \kappa}\leq \Pi_{\kappa}$. 
The proposed controller (blue) is compared to a normal CLF-CBF controller (magenta), which proposed in~\cite{ames2019control} and not consider the estimation for external disturbance. 
The corresponding simulation results are shown in Figure~\ref{fig:subfig1},~\ref{fig:subfig2} and~\ref{fig:subfig3}. We can see the system~\eqref{sys example} reaches the input constraint around $t=1.5, 6.0$, and $7.0$ seconds, where nominal control input leaves the safe set. The proposed method remains feasible and safe for the entire duration, by applying brakes early,
around $t = 6.5$ seconds, instead of $t = 6.0$ seconds. 
\begin{figure*}[!t]
  \centering
  \subfloat[Trajectory of $x$]{%
    \includegraphics[width=0.32\textwidth]{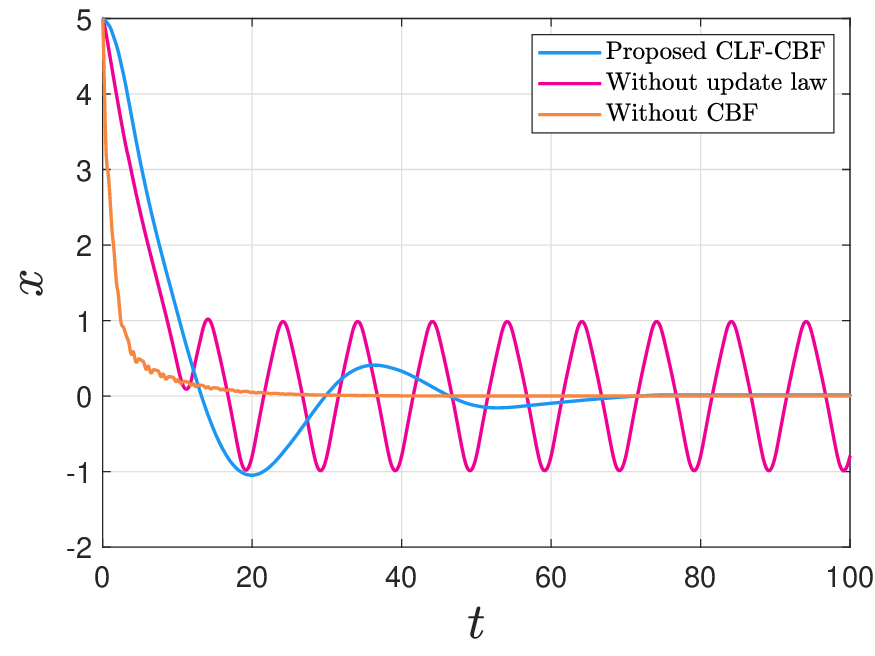}%
    \label{fig:subfig4}%
  }\hfill
  \subfloat[Control input $u$]{%
    \includegraphics[width=0.32\textwidth]{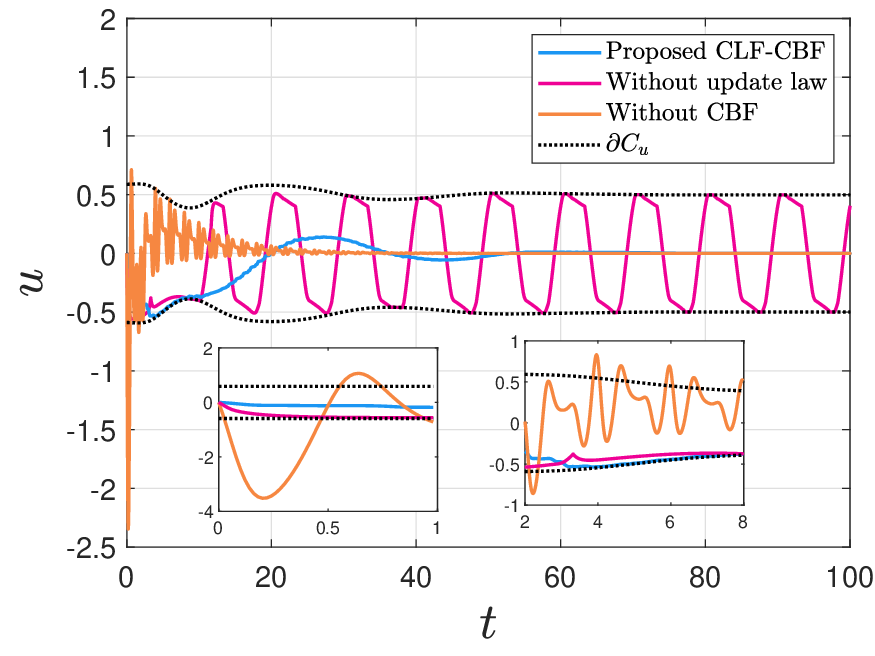}%
    \label{fig:subfig5}%
  }\hfill
  \subfloat[Barrier function $h$]{%
    \includegraphics[width=0.32\textwidth]{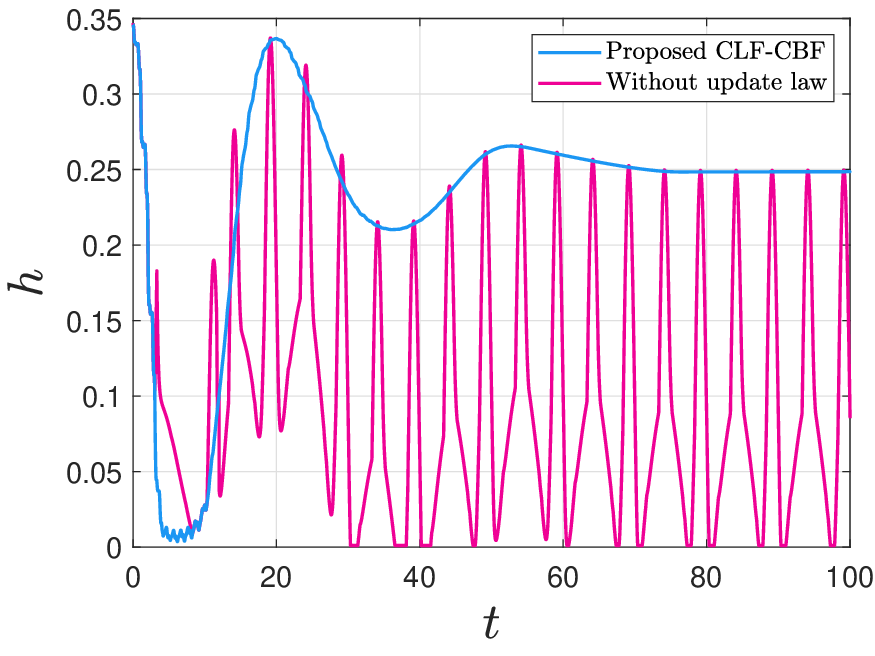}%
    \label{fig:subfig6}%
  }
  \vspace{-1.5mm}
  \caption{Case 2 simulation results. Subfigures (a), (b), and (c) compare the system state trajectory $x$, control input $u$, and barrier function $h$, respectively, obtained using the proposed CLF-CBF controller (blue), a CLF-CBF controller without the update law, and a CLF controller with the update law.}
  \label{fig:case2}
\end{figure*}
In the second numerical study, we consider a planar single-integrator system with external disturbance by letting $f(x) = 0$, $g(x)=1$ in~\eqref{sys1 casestudy}. We set the time-varying disturbances as
\begin{equation}d_x(t)=d_u(t)=\begin{cases}\frac{d_{\max}}{2}t,\quad0\leq t<\frac{T}{6},\\d_{\max}t,\quad\frac{T}{6}\leq t<\frac{T}{3},\\\frac{d_{\max}}{2}(\frac{T}{2}-t),\quad\frac{T}{3}\leq t<\frac{2T}{3},\\-d_{\max},\quad\frac{2T}{3}\leq t<\frac{5T}{6},\\\frac{d_{\max}}{2}(t-T),\quad\frac{5T}{6}\leq t\leq T,\end{cases}\end{equation}
and the maximum amplitude of the disturbance $d_{\max} = 1$.
 We set the system initial conditions as $x(0)=5, u(0)=0$. The positive constants in the simulation are selected as $c_x = c_u =0.21$, $\theta_x = \theta_u = 0.1$, $\varrho = 0.95$ and $\Pi_\kappa = 15$. Other parameters in this simulation are selected as $\nu = 0.1$, $l = 5$, $\bar{d}_i = 20$, $T = 120$s and $\lambda_x = \lambda_u = 1$.

We intend to control the system to an equilibrium point $\lim_{t\to\infty}x(t)=0$ with a state and time-related barrier function which follows the definition in~\eqref{ic} and~\eqref{bf}, and we further define
 $\kappa  = \left(-0.1\sin(x)-1/(t+10)+0.25\right)^{\frac{1}{2}}$.
 Then our proposed controller (blue) for system~\eqref{sys2 casestudy} is adopted by solving the QP problem~\eqref{equ qp1}, \eqref{equ qp2}
where the weights $\hat{w}_{h,j}$, $\hat{w}_{x,i}$ and $\hat{w}_{u,i}$ are updated by~\eqref{hat di dot} and~\eqref{adaptive and update}. 
We compared the proposed controller with the normal CLF-CBF controller (magenta) proposed
in~\cite{ames2019control}, and only using the nominal controller in~\eqref{nominal phi} without using CBF (orange). 
The simulation results are shown in Figure~\ref{fig:subfig4},~\ref{fig:subfig5} and~\ref{fig:subfig6}. The system approaches the input constraint from $t=2.0$ to $8.0$ seconds, where nominal control input leaves the safe set. In contrast, the proposed CLF-CBF method ensures the safety of the input-constrained system for all $t\geq 0$.
\section{Conclusion}\label{section-6}
The novel input-constrained CBF scheme in this paper effectively addresses the challenges of controlling full-state and input-constrained nonlinear systems. By employing an input-to-output auxiliary transformation, the original input constraints are converted into an output CBF design, thus bypassing the limitations imposed by the constraints. 
Simulation results validate the algorithm's effectiveness. 
Future research could explore the ``anti-windup" problem associated with the proposed CBF-based input constraints~\cite{tarbouriech2009anti, tarbouriech2011stability,hippe}, and refine the algorithm for specific applications~\cite{mir2} in real-world scenarios.
\bibliographystyle{ieeetr}

\end{document}